\documentclass[envcountsame]{llncs}
\pdfoutput=1
\usepackage{amsmath,amssymb}
\usepackage{caption}
\usepackage{hyperref}
\usepackage{cleveref}
\usepackage{algorithm}
\usepackage[noend]{algpseudocode}
\usepackage{enumerate}
\usepackage{versions}
\usepackage{mathtools} 
\usepackage{pdflscape}

\usepackage[american]{babel}
\usepackage[babel]{csquotes}
\usepackage[maxcitenames=2,maxbibnames=9,backend=bibtex,doi=false,isbn=false,url=false]{biblatex}
\bibliography{references}

\excludeversion{SHORT}
\includeversion{LONG}
\excludeversion{INUTILE}
\excludeversion{TODO}

\DeclarePairedDelimiter\floor{\lfloor}{\rfloor}

\newcommand{\bigo}[1]{\mathcal{O}(#1)}

\newcommand{\tuple}[1]{( #1 )}
\newcommand{\tuplelr}[1]{\left( #1 \right)}

\begin{document}

\hyphenation{pro-blems}
\pagestyle{headings} 

\mainmatter 
\title{ Depth Distribution in High Dimensions \\{\small (Extended version\thanks{A preliminary version of these results were presented at the $23^{\text{rd}}$ Annual International Computing and Combinatorics Conference (COCOON'17)\cite{2017-COCOON-DepthDistributionInHighDimension-BabrayPerezRojas}})}}

\author{
 J\'er\'emy Barbay \inst{1}
\and
Pablo P\'erez-Lantero \inst{2}
 \and 
 Javiel Rojas-Ledesma\inst{1}
}

\institute{
 Departmento de Ciencias de la Computaci\'on, Universidad de Chile, Chile\\ 
 \email{jeremy@barbay.cl, jrojas@dcc.uchile.cl.}
\and
Departmento de Matem\'atica y  Computaci\'on, Universidad de Santiago, Chile.\\ 
\email{pablo.perez.l@usach.cl.}
}



\maketitle

\begin{abstract} 
Motivated by the analysis of range queries in databases, we introduce the computation of the \textsc{Depth Distribution} of a set $\mathcal{B}$ of axis aligned boxes, whose computation generalizes that of the \textsc{Klee's Measure} and of the \textsc{Maximum Depth}. In the worst case over instances of fixed input size $n$, we describe an algorithm of complexity within $\bigo{n^\frac{d+1}{2}\log n}$, using space within $\bigo{n\log n}$, mixing two techniques previously used to compute the \textsc{Klee's Measure}. We refine this result and previous results on the \textsc{Klee's Measure} and the \textsc{Maximum Depth} for various measures of difficulty of the input, such as the profile of the input and the degeneracy of the intersection graph formed by the boxes.
\end{abstract}


\section{Introduction}\label{sec:Introduction}				 

Problems studied in Computational Geometry have found important applications in the processing and querying of massive databases~\cite{AboKhamis2015}, such as the computation of the \textsc{Maxima} of a set of points~\cite{Afshani14,2017-JACM-InstanceOptimalGeometricAlgorithms-AfshaniBarbayChan}, or compressed data structures for \textsc{Point Location} and \textsc{Rectangle Stabbing}~\cite{Afshani2012}.
In particular, we consider cases where the input or queries are composed of axis-aligned boxes in $d$ dimensions: in the context of databases it corresponds for instance to a search for cars within the intersection of ranges in price, availability and security ratings range. 

Consider a set $\mathcal{B}$ of $n$ axis-parallel boxes in $\mathbb{R}^d$, for fixed $d$.
\begin{LONG}
We focus on two measures on such set of boxes: the \textsc{Klee's measure} and the \textsc{Maximum Depth}.
\end{LONG}
The \textsc{Klee's Measure} of $\mathcal{B}$ is
\begin{LONG}
 the size of the ``shadow'' projected by $\mathcal{B}$, more formally 
\end{LONG}
the volume of the union of the boxes in $\mathcal{B}$. 
Originally suggested on the line by Klee~\cite{Klee1977}, its computation is well studied in higher dimensions~\cite{bringmann2012,bringmann2013,Chan08,Chan2013,Overmars1991,Yildiz11,YildizHershbergerSuri11,Yildiz2012}, and can be done in time within $\bigo{n^{d/2}}$, using an algorithm introduced by Chan~\cite{Chan2013} based on a new paradigm called ``Simplify, Divide and Conquer".
The \textsc{Maximum Depth} of $\mathcal{B}$ is the maximum number of boxes covering a same point, and its computational complexity is similar to that of \textsc{Klee's Measure}'s, converging to the same complexity within $\bigo{n^{d/2}}$~\cite{Chan2013}.

\noindent\emph{Hypothesis.}
The known algorithms to compute these two measures are all strikingly similar\begin{LONG}, to the point that Chan~\cite{Chan2013} states that all known techniques used for computing the \textsc{Klee's Measure} can be applied to the computation of the \textsc{Maximum Depth}\end{LONG}.
That would suggest a reduction from one to the other, but those two measures are completely distinct: the  \textsc{Klee's measure} is a volume whose value can be a real number, while the \textsc{Maximum Depth} is a cardinality whose value is an integer in the range $[1..n]$.
\textbf{Is there any way to formalize the close relationship between the computation of these two measures?}

\noindent\emph{Our Results.}
We describe a first step towards such a formalization, in the form of \textbf{a new problem}, which we show to be intermediary in terms of the techniques being used, between the \textsc{Klee's Measure} and the \textsc{Maximum Depth}, slightly more costly in time and space, and with interesting applications and results of its own.
%
We introduce the notion of \textsc{Depth Distribution} of a set $\mathcal{B}$ of $n$ axis-parallel boxes in $\mathbb{R}^d$, formed by the vector of $n$ values $(V_1,\ldots,V_n)$, where $V_i$ corresponds to the volume covered by exactly $i$ boxes from $\cal B$.
The \textsc{Depth Distribution} of a set  $\mathcal{B}$  can be interpreted as a probability distribution function (hence the name): if a point $p$ is selected uniformly at random from the region covered by the boxes in $\mathcal{B}$, the probability that $p$ hits exactly $k$ boxes from $\mathcal{B}$ is $\left( {V_k} \middle/ {\sum_{i=1}^{n}{V_i}} \right)$, for all $ k \in [1..n]$.

The \textsc{Depth Distribution} \textbf{refines both} the \textsc{Klee's Measure} and the \textsc{Maximum Depth}. 
It is a measure finer than the \textsc{Klee's Measure} in the sense that the \textsc{Klee's Measure} of a set $\cal B$ can be obtained in time linear in the size $n$ of $\cal B$ by summing the components of the \textsc{Depth Distribution} of $\cal B$.
Similarly, the \textsc{Depth Distribution} is a measure finer than the \textsc{Maximum Depth} in the sense that the \textsc{Maximum Depth} of a set $\cal B$ can be obtained in linear time by finding the largest $i\in[1..n]$ such that $V_i\neq0$. 
%
In the context of a database, when receiving multidimensional range queries (e.g. about cars), the \textsc{Depth Distribution} of the queries yields valuable information to the database owner (e.g. a car dealer) about the repartition of the queries in the space of the data, to allow for informed decisions on it (e.g. to orient the future purchase of cars to resell based on the clients' desires, as expressed by their queries).

In the classical computational complexity model where one studies the worst case over instances of fixed size $n$,
\begin{LONG}
	the trivial approach of partitioning the space in cells that are completely contained within all the boxes they intersect, results in a solution with prohibitive running time (within $\bigo{n^{d+1}}$). Simple variants of the techniques previously used to compute the \textsc{Klee's Measure}~\cite{Chan2013,Overmars1991} results in a solution running in time within $\bigo{n^{d/2+1}}$, using linear space, or a solution running in time within $\bigo{n^{(d/2+1/2}\log n}$, but using space within $\bigo{n^{d/2}\log n}$.
	We combine those two into a single technique to compute the \textsc{Depth Distribution} in time within $\bigo{n^\frac{d+1}{2}\log n}$, using space within $\bigo{n\log n}$ (in Section~\ref{sec:classical-worst-case})
\end{LONG}
\begin{SHORT}
	we combine techniques previously used to compute the \textsc{Klee's Measure}~\cite{Chan2013,Overmars1991} to compute the \textsc{Depth Distribution} in time within $\bigo{n^\frac{d+1}{2}\log n}$, using space within $\bigo{n\log n}$ (in Section~\ref{sec:classical-worst-case}).
\end{SHORT}
This solution is slower by a factor within $\bigo{\sqrt{n} \log n}$ than the best known algorithms for computing the \textsc{Klee's Measure} and the \textsc{Maximum Depth}: we show in Section~\ref{sec:lower_bounds} that such a gap might be ineluctable, via a reduction from the computation of \textsc{Matrix Multiplication}.

In the refined computational complexity model where one studies the worst case complexity taking advantage of \textbf{additional parameters describing the difficulty} of the instance~\cite{2017-JACM-InstanceOptimalGeometricAlgorithms-AfshaniBarbayChan,Kirkpatrick1985,MoffatP92}, we consider (in \Cref{sec:finer_analysis}) distinct measures of difficulty for the instances of these problems, such as the \emph{profile} and the \emph{degeneracy} of the intersection graph of the boxes, and describe algorithms in these model to compute the \textsc{Depth Distribution}, the \textsc{Klee's Measure} and the \textsc{Maximum Depth} of a set $\cal B$.
 \begin{INUTILE}
 Javiel: can we use those results on polynome multiplication too?
 \end{INUTILE}

 After a short overview of the known results on the computation of the \textsc{Klee's Measure} and the \textsc{Maximum Depth} (in \Cref{sec:preliminaries}), we describe in Section~\ref{sec:computing_depthdist} the results in the worst case over instances of fixed size. In Section~\ref{sec:finer_analysis}, we describe results on refined partitions of the instance universe, both for the computation of the \textsc{Depth Distribution}, and for the previously known problems of computing the \textsc{Klee's Measure} and the \textsc{Maximum Depth}.  We conclude in Section~\ref{sec:discussion} with a discussion on discrete variants and further refinements of the analysis.

\section{Background}
\label{sec:preliminaries}
The techniques used to compute the \textsc{Klee's Measure} have evolved over time, and can all be used to compute the \textsc{Maximum Depth}.
We retrace some of the main results, which will be useful for the definition of an algorithm computing the \textsc{Depth Distribution} (in \Cref{sec:computing_depthdist}), and for the refinements of the analysis for \textsc{Depth Distribution}, \textsc{Klee's Measure}  and \textsc{Maximum Depth} (in \Cref{sec:finer_analysis}).

The computation of the \textsc{Klee's Measure} of a set $\mathcal{B}$ of $n$ axis-aligned $d$-dimensional boxes was first posed by Klee~\cite{Klee1977} in 1977.
After some initial progresses~\cite{Bentley1977,Fredman78,Klee1977}, Overmars and Yap~\cite{Overmars1991} described  a  solution running in time within $\bigo{n^{d/2}\log n}$. This remained the best solution for more than 20 years until 2013, when Chan~\cite{Chan2013} presented a simpler and faster algorithm running in time within $\bigo{n^{d/2}}$.

The algorithms described by Overmars and Yap~\cite{Overmars1991} and by Chan~\cite{Chan2013}, respectively, take both advantage of solutions to the special case of the problem where all the boxes are slabs. 
A box $b$ is said to be a \emph{slab} within another box $\Gamma$ if $b \cap \Gamma =  \{(x_1, \ldots, x_d) \in \Gamma \mid \alpha \le x_i  \le \beta\}$, for some integer $i \in [1..d]$ and some real values $\alpha,\beta$ (see \Cref{fig:slabs} for an illustration). 
Overmars and Yap~\cite{Overmars1991} showed that, if all the boxes in $\mathcal{B}$ are slabs inside the domain box $\Gamma$, then the \textsc{Klee's Measure} of $\mathcal{B}$ within $\Gamma$ can be computed in linear time (provided that the boxes have been pre-sorted in each dimension).

\begin{figure}
	\begin{center}
		\includegraphics[width=1\linewidth]{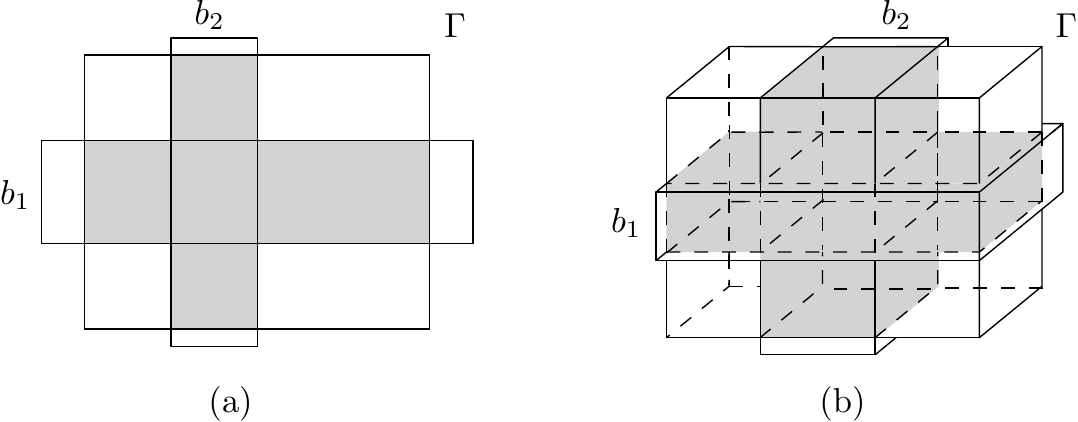}
	\end{center}   
	\caption{An illustration in dimensions 2 (a) and 3 (b) of two boxes $b_1, b_2$ equivalent to slabs when restricted to the box $\Gamma$.
		The \textsc{Klee's Measure} of $\{b_1,b_2\}$ within $\Gamma$ is the area (resp. volume) of the shadowed region in (a) (resp. (b)).} \label{fig:slabs}
\end{figure}

Overmars and Yap's algorithm~\cite{Overmars1991} is based on a technique originally described by Bentley~\cite{Bentley1977}: solve the static problem in  $d$ dimensions by combining a data structure for the dynamic version of the problem in $d-1$ dimensions with a plane sweep over the $d$-th dimension.
The algorithm starts by partitioning the space into $\bigo{n^{d/2}}$ rectangular cells such that the boxes in $\mathcal{B}$ are equivalent to slabs when restricted to each of those cells. 
Then, the algorithm builds a \emph{tree-like} data structure whose leaves are the cells of the partition, supporting insertion and deletion of boxes  while keeping track of the \textsc{Klee's Measure} of the boxes.

Chan's algorithm~\cite{Chan2013} is a simpler \emph{divide-and-conquer} algorithm, where the slabs are \emph{simplified} and removed from the input before the recursive calls (Chan~\cite{Chan2013} named this technique \emph{Simplify, Divide and Conquer}, \emph{SDC} for short).
To obtain the recursive subproblems, the algorithm assigns a constant weight of $2^{\frac{i+j}{2}}$ to each ($d$-2)-face intersecting the domain and orthogonal to the $i$-th and $j$-th dimensions, $i,j \in [1..d]$.
Then, the domain is partitioned into two sub-domains by the hyperplane $x_1 = m$, where $m$ is the weighted median of the ($d$-2)-faces orthogonal to the first dimension. This yields a decrease by a factor of $2^{2/d}$ in the total weight of the ($d$-2)-faces intersecting each sub-domain.
Chan~\cite{Chan2013}  uses this, and the fact that slabs have no ($d$-2)-face intersecting the domain, to prove that the SDC algorithm runs in time within $\bigo{n^{d/2}}$.
%

Unfortunately, there are sets of boxes which require partitions of the space into a number of cells within $\Omega(n^{d/2})$ to ensure that, when restricted to each cell, all the boxes are equivalent to slabs. 
Hence, without a radically new technique, any algorithm based on this approach will require running time within $\Omega(n^{d/2})$.
Chan~\cite{Chan2013} conjectured that any \emph{combinatorial} algorithm computing the \textsc{Klee's Measure} requires within $\Omega(n^{d/2})$ operations, via a reduction from the parameterized \textsc{k-Clique} problem, in the worst case over instances of fixed size~$n$. 
As a consequence, recent work have focused on the study of special cases which can be solved faster than $\Omega(n^{d/2})$, like for instance when all the boxes are \emph{orthants}~\cite{Chan2013}, \emph{$\alpha$-fat boxes}~\cite{bringmann2012}, or cubes~\cite{bringmann2013}.
In turn, we show in \Cref{sec:finer_analysis} that there are measures which gradually separate \emph{easy} instances for these problems from the \emph{hard} ones.

In the next section, we present an algorithm for the computation of the \textsc{Depth Distribution} inspired by a combination of the approaches  described above, outperforming naive applications of those techniques.

\section{Computing the Depth Distribution}
\label{sec:computing_depthdist}
We describe in  \Cref{sec:classical-worst-case} an algorithm to compute the \textsc{Depth Distribution} of a set of $n$ boxes.
The running time of this algorithm in the worst case over $d$-dimensional instances of fixed size $n$ is within $\bigo{n^{\frac{d+1}{2}}\log n}$, using space within $\bigo{n \log n}$.
This running time is worse than that of computing only the \textsc{Klee's Measure} (or the \textsc{Maximum Depth}) by a factor within $\bigo{\sqrt{n} \log n}$: we argue in \Cref{sec:lower_bounds} that computing the \textsc{Depth Distribution} is computationally harder than the special cases of computing the \textsc{Klee's Measure} and the \textsc{Maximum Depth}, unless computing \textsc{Matrix Multiplication} is much easier than usually assumed.
%

\subsection{Upper bound}\label{sec:classical-worst-case}
We introduce an algorithm to compute the \textsc{Depth Distribution} inspired by a combination of the techniques introduced  by Chan~\cite{Chan2013}, and by Overmars and Yap~\cite{Overmars1991}, for the computation of the \textsc{Klee's Measure} (described in \Cref{sec:preliminaries}).
As in those approaches, the algorithm partitions the domain $\Gamma$ into $\bigo{n^{d/2}}$ cells where the boxes of $\mathcal{B}$ are equivalent to slabs, and then  combines the solution within each cell to obtain the final answer. 
Two main issues must be addressed: how to compute the \textsc{Depth Distribution} when the boxes are slabs, and how to partition the domain efficiently. 

We address first the special case of slabs. We show in \Cref{lem:ddist_slabs} that computing the \textsc{Depth Distribution} of a set of $n$ $d$-dimensional slabs within a domain $\Gamma$ can be done via a multiplication of $d$ polynomials of degree at most $n$.

\begin{lemma}\label{lem:ddist_slabs}
	Let  $\mathcal{B}$ be a set of $n$ axis-parallel $d$-dimensional axis-aligned boxes
	whose intersection with a domain box $\Gamma$ are slabs.
	The computation of the \textsc{Depth Distribution} $\tuplelr{V_1, \ldots, V_n}$ of $\mathcal{B}$ within $\Gamma$ can be performed via a multiplication of $d$ polynomials of degree at most $n$.
\end{lemma}
\begin{proof}
	For all $i \in [1..d]$, let $\mathcal{B}_i$ be the subset of slabs that are orthogonal to the $i$-th dimension, and
	let $\tuplelr{V^i_1, \ldots, V^i_n}$ be the \textsc{Depth Distribution} of the intervals that result from projecting $\mathcal{B}_i$ to the $i$-th dimension within $\Gamma$. We associate a polynomial $P_i(x)$ of degree $n$ with each $\mathcal{B}_i$  as follows:
	\begin{itemize}
		\item let $\Gamma_i$ be the projection of the domain $\Gamma$ into the $i$-th dimension, and  
		\item let $V^i_0$ be the length of the region of $\Gamma_i$ not covered by a box in $B_i$ (i.e., $V^i_0= (|\Gamma_i| - \sum_{j=1}^{n}{V^i_j})$); then
		\item $P_i(x) = \sum_{j=0}^{n}{V^i_j} \cdot x^j$.
	\end{itemize}
	Since any slab entirely covers the domain in all the dimensions but the one to which it is orthogonal, any point $p$ has depth $k$ in $\mathcal{B}$ if and only if it has depth $j_1$ in $\mathcal{B}_1$, $j_2$ in $\mathcal{B}_2$, $\ldots,$ and $j_d$ in  $\mathcal{B}_d$, such that $j_1 + j_2 + \ldots + j_d = k$. Thus, for all $k \in [0..n]$:
	\begin{equation*}\label{eq:reduction}
	V_k = \sum_{\substack{0 \le j_1, \ldots, j_d \le n \\ j_1 + \ldots + j_d = k}}{\left(\prod_{i=1}^{d} V^i_{j_i}\right)},
	\end{equation*}
	which is precisely the $(k+1)$-th coefficient of $P_1(x) \cdot P_2(x) \cdot \ldots \cdot P_d(x)$. Thus,
	this product yields the \textsc{Depth Distribution} $\tuplelr{V_1, \ldots, V_n}$ of $\mathcal{B}$ in  $\Gamma$.
	\qed
\end{proof}

\begin{INUTILE}
\begin{lemma}
Let  $\mathcal{B}$ be a set of $n$ axis-parallel $d$-dimensional axis aligned boxes, with $d \ge 2$, that are equivalent to slabs when restricted to a domain box $\Gamma$.
Then, the problem of computing the \textsc{Depth Distribution} of $\mathcal{B}$ is equivalent to the problem of multiplying $d$ polynomials of degree at most $n$.
\end{lemma}
\begin{proof}
To prove that the problems are equivalent we show that \textsc{Depth Distribution} is a subproblem of \textsc{Polynomial Multiplication}, and \emph{vice versa}.
\\
\\
\textsc{Depth Distribution} $\subseteq$ \textsc{Polynomial Multiplication}:\\
For all $i$ in $[1..n]$ let:
\begin{itemize}
\item $\mathcal{B}_i$ be the slabs that are orthogonal $x_i$;
\item $\tuplelr{V^i_1, \ldots, V^i_n}$ be the \textsc{Depth Distribution} of the projection of $\mathcal{B}_i$ to $x_i$;
\item $|\Gamma_i|$ denote the length of $\Gamma$ when projected to dimension $x_i$;
\item $V^i_0 = (|\Gamma|_i - \sum_{j=1}^{n}{V^i_j})$; and
\item $P_i(x) = \sum_{j=0}^{n}{V^i_j} \cdot x^j$  be a polynomial of degree $n$.
\end{itemize}

Any slab covers the entire domain in all the dimensions except the one to which it is orthogonal.
Therefore, any point $p$ has depth $k$ in $\mathcal{B}$ if it has depth $j_1,\ldots,j_d$ in  $\mathcal{B}_1,\ldots,\mathcal{B}_d$, respectively, such that $j_1 + \ldots + j_d = k$.
Thus, for all $k=[0..n]$:
\begin{equation}
V_k = \sum_{\substack{0 \le j_1, \ldots, j_d \le n \\ j_1 + \ldots + j_d = k}}{\left(\prod_{i=1}^{d} V^i_{j_i}\right)},
\end{equation}
which is precisely the $k$-th coefficient of $P_1(x) \times \ldots \times P_d(x)$.
So computing the product yields the \textsc{Depth Distribution}.
\\\\
\textsc{Polynomial Multiplication} $\subseteq$ \textsc{Depth Distribution}:\\
Let $P_i(x) = \sum_{j=0}^{n}{a^i_j} x^j$, for $i=[1..n]$ be $d$ polynomials of degree $n$.
We show that there is a set of boxes $\mathcal{B}$, such that from the \textsc{Depth Distribution} of $\mathcal{B}$, one obtains with little extra effort the coefficients of $P_1(x) \times \ldots \times P_d(x)$.

\begin{figure}[t]
	\begin{minipage}[c]{.6\textwidth}
		\centering
		\includegraphics[scale=.9]{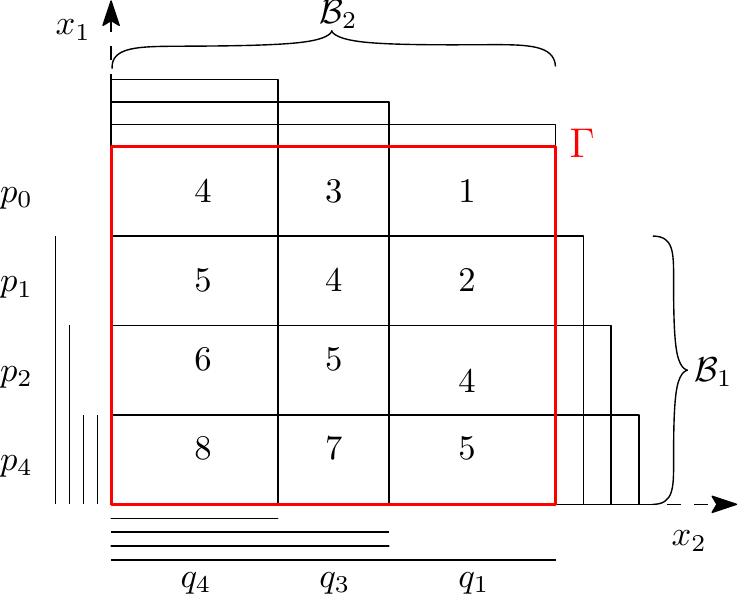}
	\end{minipage}\hfill
	\begin{minipage}[c]{.4\textwidth}
		\captionof{figure}{Representation of the product $P(n) \times Q(n)$ as an instance of \textsc{Depth Distribution}, for $P(n)=p_4  n^4 + p_2  n^2 + p_1 n + p_0 $, and $Q(n)=q_4  n^4 + q_3  n^3 + q_1 n $.$\mathcal{B}_1'$ and $\mathcal{B}_2'$ represent the sets of projections of the boxes created for $P$ and $Q$ to dimensions $x_1$ and $x_2$, respectively.
			The number within each cell indicates the depth of the corresponding region.}
		\label{fig:reduction}
	\end{minipage}
\end{figure}

We construct a domain box $\Gamma$, and for the $i$-th polynomial, we create a set of slabs within $\Gamma$ orthogonal to the $i$-th dimension, as follows (see \Cref{fig:reduction} for an illustration):
\begin{itemize}
	\item $\Gamma = \{(x_1, \ldots, x_d) \mid 0 \le x_1 \le \sum_{j=0}^{n}{a^1_j} \wedge \ldots \wedge 0 \le x_d \le \sum_{j=0}^{n}{a^d_j}\} $
	(in the $i$-th dimension, $\Gamma$ extends from zero to the sum of the coefficients of $P_i$, for $i=[1..d]$);
	
	\item For all $i=[1..d]$, $\mathcal{B}_i = \bigcup_{j=1}^{n} \left(\Gamma \cap \{(x_1, \ldots, x_d) \in \Gamma \mid 0 \le x_i \le \sum_{k=0}^{j}{a^i_k} \}\right)$ ($\mathcal{B}_i$ has $n$ boxes such that the length of the region in the $i$-th coordinate covered by exactly $k$ boxes is $a^i_k$, for $k=[0..n]$); and
	
	\item $\mathcal{B} = \bigcup_i^d \mathcal{B}_i$ ($\mathcal{B}$ is the set of $d \cdot n$ boxes resulting from the union of $\mathcal{B}_i$, for $i=[1..d]$).
\end{itemize}

When constructed as above, the \textsc{Depth Distribution} of the one dimensional projection of $\mathcal{B}_i$ to the $i$-th dimension is precisely $\tuple{V^i_1=a^i_1, \ldots, V^i_n=a^i_n}$ for all $i=[1..d]$, and the \textsc{Depth Distribution} of $\mathcal{B}$ is given by \Cref{eq:reduction}.
Thus, computing the \textsc{Depth Distribution} of $\mathcal{B}$ yield the first $(nd-1)$ coefficients of $P_1(x) \times \ldots \times P_n(x)$, in descending order of the exponents.
Finally, the last coefficient of the product (i.e. the constant coefficient) is given by $|\Gamma| - \sum_{i=1}^{n}{V_i}$, where $|\Gamma|$ denotes the volume of $\Gamma$.
\qed
\end{proof}
\end{INUTILE}

Using standard \emph{Fast Fourier Transform} techniques, two polynomials can be multiplied in time within $O(n \log n)$~\cite{CormenLRS-ItoA-2009}.
Moreover, the \textsc{Depth Distribution} of a set of intervals (i.e., when $d=1$) can be computed in linear time after sorting, by a simple scan-line algorithm, as for the \textsc{Klee's Measure}~\cite{Chan2013}.
Thus, as a consequence of Lemma~\ref{lem:ddist_slabs}, when the boxes in $\mathcal{B}$  are slabs when restricted to a domain box $\Gamma$, the \textsc{Depth Distribution} of $\mathcal{B}$ within $\Gamma$ can be computed in time within $\bigo{n \log n}$. 

\begin{corollary}\label{col:ddist_slabs}
  Let $\mathcal{B}$ be a set of $n$ $d$-dimensional axis aligned boxes whose intersections with a $d$-dimensional box $\Gamma$ are slabs. The \textsc{Depth Distribution} of $\mathcal{B}$ inside $\Gamma$ can be computed in time within $\bigo{n \log n}$.
\end{corollary}

\begin{INUTILE}
With arguments similar to the ones used in the proof of \Cref{lem:ddist_slabs} one can show that multiplying two polynomials of degree $n$ can be reduced in linear time to an instance of the computation of the \textsc{Depth Distribution} of a set of $\bigo{n}$ boxes. Hence, the running time of \Cref{col:ddist_slabs} is tight.	
\end{INUTILE}

A naive application of previous techniques~\cite{Chan2013,Overmars1991} to the computation of \textsc{Depth Distribution} yields poor results. Combining the result in \Cref{col:ddist_slabs} with the partition of the space and the data structure described by Overmars and Yap~\cite{Overmars1991} yields an algorithm to compute the \textsc{Depth Distribution} running in time within  $\bigo{n^{\frac{d+1}{2}}\log n}$, and using space within $\bigo{n^{d/2}\log n}$.
Similarly, if the result in \Cref{col:ddist_slabs} is combined with Chan's partition of the space~\cite{Chan2013}, one obtains an algorithm using space linear in the number of boxes,  but running in time within $\bigo{n^{\frac{d}{2}+1}\log n}$ (i.e., paying an extra $\bigo{n^\frac{1}{2}}$-factor for the reduction in space usage of~Overmars and Yap~\cite{Overmars1991}). 


We combine these two approaches into an algorithm which achieves the best features of both: it runs in time within $\bigo{n^\frac{d+1}{2}\log n}$, and uses $\bigo{n \log n}$-space. As in Chan's approach~\cite{Chan2013} we use a divide and conquer algorithm, but we show in \Cref{theo:ddist_comp_spacey} that the running time is asymptotically the same as if using the partition and data structures described by Overmars and Yap~\cite{Overmars1991}  (see \Cref{alg:ddistribution-inner} for a detailed description). 

\begin{INUTILE}
		We show how to transform any set of $n$ boxes $\mathcal{B}$ into a $\mathcal{B'}$ with $\bigo{n}$ boxes, such that any rectangular region which does not intersect a ($d$-2)-face of a box in $\mathcal{B'}$, partially intersects (i.e., intersects but does not contain) at most $\bigo{\sqrt{n}}$ boxes of $\mathcal{B}$. 
		
		\begin{lemma}
			Let  $\mathcal{B}$ be a set of $n$ axis-parallel boxes in $\mathbb{R}^d$, and let $\Gamma$ be a $d$-dimensional domain box.
			There exists a set $\mathcal{B}'$  such that:
			\begin{enumerate}[a)]
				\item{the \textsc{Depth Distribution} of $\mathcal{B}'$ and $\mathcal{B}$ within $\Gamma$ are the same;}
				\item{the size of $\mathcal{B}'$ is within $\bigo{n}$}
				\item{any box $\gamma$ which does not intersect a ($d$-2)-face of some box in $\mathcal{B}'$, partially intersects more than $\sqrt{n}$ boxes from $\mathcal{B}'$ in at most one dimension.}
			\end{enumerate}
		\end{lemma}
		
		\begin{proof}
			Initially, let $\mathcal{B}'=\mathcal{B}$. We add to $\mathcal{B'}$ \emph{empty} boxes of the form $\{(x_1, \ldots, x_d) \in \Gamma \mid x_i = \alpha, x_j=\beta\}$ for some $i,j \in [1..d]$, and some endpoints $\alpha, \beta$ of the boxes of $\mathcal{B}$. 
			For this, for every pair of dimensions $i,j$, let $I$ (resp. $J$) be the set of endpoints at positions $\frac{\sqrt{n}}{d}, \frac{2\sqrt{n}}{d}, \ldots, \frac{d\sqrt{n}\sqrt{n}}{d}$  in the $i$-th dimension (resp. the $j$-dimension) in ascending order. For  each pair $(\alpha \in I, \beta \in J)$ we add the box $\{(x_1, \ldots, x_d) \in \Gamma \mid x_i = \alpha, x_j=\beta\}$ to $\mathcal{B'}$. 
			As the boxes added to $\mathcal{B}'$ are empty (i.e., their volume is zero), the \textsc{Depth Distribution} of $\mathcal{B}'$ is the same as that of $\mathcal{B}$. The total number of boxes added for each pair of dimensions is $2d\sqrt{n} \times 2d\sqrt{n} = 4d^2 n$ boxes, and because the number of distinct pairs of dimensions is constant, the size of $\mathcal{B}'$ is within $\bigo{n}$. The third assertion derives directly from the construction.
			\qed
		\end{proof}
		
		The result in \Cref{lem:ddist_slabs} allows to use the partition scheme introduced by Chan~\cite{Chan2013}, while obtaining the nice properties yielded by the partition scheme introduced by Overmars and Yap~\cite{Overmars1991} (which is key to the algorithm we describe below for computing the \textsc{Depth Distribution} of a set of boxes).
\end{INUTILE}

\begin{algorithm}
	\caption{\texttt{SDC-DDistribution}($\mathcal{B}, \Gamma, c, \tuplelr{V_1, \ldots, V_n}$)}
		\label{alg:ddistribution-inner}          
	\begin{algorithmic}[1]
		\Require A set $\mathcal{B}$ of $n$ boxes in $\mathbb{R}^d$; a $d$-dimensional domain box $\Gamma$; the number $c$ of boxes not in $\mathcal{B}$ but in the original set that completely contain $\Gamma$; and a vector $\tuplelr{V_1, \ldots, V_n}$ representing the \textsc{Depth Distribution} computed so far.
		\If {no box in $\mathcal{B}$ has a ($d$-2)-face intersecting $\Gamma$ (i.e., all the boxes are slabs)} \label{step:basecase_slabs}
		\State Compute the \textsc{Depth Distribution} $\tuplelr{V'_1, \ldots, V'_{|\mathcal{B}|}}$ of $\mathcal{B}$ within $\Gamma$ using \Cref{lem:ddist_slabs}
		\For {$i \in [1..|\mathcal{B}|]$}
		\State $V_{i+c} \gets V_{i+c} +  V'_{i} $
		\EndFor
		\Else 
		\State Let 	$\mathcal{B}^0 \subseteq \mathcal{B}$  be the subset of boxes completely containing $\Gamma$
		\State $c \gets c + |\mathcal{B}^0|$
		\State Let $\mathcal{B}' = \mathcal{B} \setminus \mathcal{B}^0$
		\State Let $m$ be the weighted median of the ($d$-2)-faces orthogonal to $x_1$;\label{step:partition_1}
		\State Split $\Gamma$ into $\Gamma_L, \Gamma_R$ by the hyperplane $x_1 = m$;
		\State Rename the dimensions so that $x_1, \ldots, x_d$ becomes $x_2, \ldots, x_d, x_1$;\label{step:partition_3}
		\State Let $\mathcal{B}_L$ and $\mathcal{B}_R$ be the subsets of  $\mathcal{B}' $ intersecting $\Gamma_L$ and $\Gamma_R$ respectively;
		\State Call \texttt{SDC-DDistribution}($\mathcal{B_L}, \Gamma_L, c, \tuplelr{V_1, \ldots, V_n}$)
		\State Call \texttt{SDC-DDistribution}($\mathcal{B_R}, \Gamma_R, c, \tuplelr{V_1, \ldots, V_n}$)
		\EndIf
	\end{algorithmic}
\end{algorithm}

\begin{theorem}\label{theo:ddist_comp_spacey} 
Let  $\mathcal{B}$ be a set of $n$ axis-parallel boxes in $\mathbb{R}^d$.
The \textsc{Depth Distribution} of $\mathcal{B}$ can be computed in time within $\bigo{n^\frac{d+1}{2}\log n}$, using space within $\bigo{n\log n}$.
\end{theorem}
\begin{SHORT}
Due to lack of space we defer the complete proof to the extended version~\cite{2017-ARXIV-DepthDistributionInHighDimension-BabrayPerezRojas}.
\end{SHORT}
\begin{LONG}
\begin{proof}
	First, let us show that the running time $T(n)$ of \Cref{alg:ddistribution-inner} is within $\bigo{n^\frac{d+1}{2}\log n}$. 
	We can charge the number of boxes in the set to the number of ($d$-1)-faces intersecting the domain: if a box in $\mathcal{B}$ does not have a ($d$-1)-face intersecting the domain, then it covers the entire domain, and it would have been removed (simplified) from the input in the parent recursive call. 
	Note that the ($d$-1)-faces orthogonal to dimension $x_1$ cannot intersect both the sub-domains $\Gamma_L$ and $\Gamma_R$ of the recursive calls at the same time (because the algorithm uses a hyperplane orthogonal to $x_1$ to split the domain into $\Gamma_L$ and $\Gamma_R$). Hence, although at the $d$-th level of the recursion there are $2^d$ recursive calls, any ($d$-1)-face can appear in at most $2^{d-1}$ of those. In general, for any $i$, there are at most  $2^{i}$ recursive calls at the $i$-th level of recursion, but any ($d$-1)-face of the original set can intersect at most $2^{\floor{i/d}(d-1)}$ of the cells corresponding to the domain of those calls. Hence, the total number of ($d$-1)-faces which \emph{survive} until the $i$-th level of the recursion tree is within $\bigo{n2^{\floor{i/d}(d-1)}}$ (a similar argument was used by Overmars and Yap~\cite{Overmars1991} to bound the running time of the data structure they introduced). 
	
	Let $h$ be the height of the recursion tree of \Cref{alg:ddistribution-inner}. Chan~\cite{Chan2013} showed that the partition in steps \ref{step:partition_1}-\ref{step:partition_3} makes $h$ to be within a constant term of $\frac{d}{2}\log n$.  We analyze separately  the total cost $T_I(n)$ of the \emph{interior nodes} of the recursion tree (i.e. the nodes corresponding to recursive calls which fail the base case condition in step \ref{step:basecase_slabs}), from the total cost $T_L(n)$ of the \emph{leaves} of the recursion tree. 
		
	Since, the cost of each interior node is linear in the number of ($d$-1)-faces intersecting the corresponding domain, $T_I(n)$ is bounded by:
	\begin{align*}
		T_I(n) &\in  \sum_{i=1}^{h}{\bigo{n2^{\floor{i/d}(d-1)}}}  \\
					 &	\subseteq  \bigo{n\sum_{i=1}^{h}{2^{i(d-1)/d}}} \\
					  &\subseteq  \bigo{n 2^{\frac{d-1}{d} h}}  \\
					  & \subseteq \bigo{n 2^{\frac{d-1}{d} \frac{d}{2}\log n }}  \\
					  & = \bigo{n^{\frac{d+1}{2}}}
	\end{align*}
	To analyze the total cost of the leaves of the recursion tree, first note that the total number $l$ of such recursive calls is within $\bigo{n^{d/2}}$. Let $n_1, \ldots, n_l$ denote the number of ($d$-1)-faces in each of those recursive calls, respectively.
	Note that $T_L(n)$ is within $\bigo{\sum_{i=1}^{l} n_i \log n_i}$ because the result of \Cref{lem:ddist_slabs} is used in line step~\ref{step:basecase_slabs} of the algorithm.
	 Besides, since the number of ($d$-1)-faces which survive until the $h$-th level of the recursion tree is within $\bigo{n^{\frac{d-1}{2}}}$ , $\sum_{i=1}^{l} {n_i} \in \bigo{n^{\frac{d+1}{2}}}$.
	 That bound, and the fact that  $\log {n_i} \le \log n$, for all $i \in [1..l]$, yields $T_L(n) \in  \bigo{n^\frac{d+1}{2}\log n}$. 	As $T(n) = T_I(n) + T_L(n)$, the bound for the running time follows. 
	
	With respect to the space used by the algorithm, note that only one path in recursion tree is active at any moment, and that at most $\bigo{n}$ extra space is needed within each recursive call. Since the height of the recursion tree is within $\bigo{\log n}$, the total space used by the algorithm is clearly within $\bigo{n \log n}$.
	\qed
\end{proof}
\end{LONG}

Note that in \texttt{SDC-DDistribution} the \textsc{Depth Distribution} is accumulated into a parameter. This is only to simplify the description and analysis of  the algorithm, it does not impact its computational or space complexity. The initialization of the parameters of \texttt{SDC-DDistribution} should be done as below:

\begin{algorithm}
	\caption{\texttt{DDistribution($\mathcal{B}, \Gamma$)}}          
	\label{alg:ddistribution}
	\begin{algorithmic}[1]
		\Require A set $\mathcal{B}$ of $n$ boxes in $\mathbb{R}^d$, a $d$-dimensional domain box $\Gamma$
		\Ensure The \textsc{Depth Distribution} of $\mathcal{B}$ within $\Gamma$
		\State $\tuplelr{V_1, V_2, \ldots, V_n} \gets \tuplelr{0,0, \ldots, 0}$
		\State \texttt{SDC-DDistribution}($\mathcal{B}, \Gamma, 0, \tuplelr{V_1, V_2, \ldots, V_n}$)
		\State  return $\tuplelr{V_1, V_2, \ldots, V_n}$
	\end{algorithmic}
\end{algorithm}

The bound for the running time in \Cref{theo:ddist_comp_spacey} is worse than that of computing the \textsc{Klee's Measure} (and \textsc{Maximum Depth}) by a factor within $\bigo{\sqrt{n} \log n}$, which raises the question of the optimality of the bound: we consider this matter in the next section.

\subsection{Conditional Lower Bound}\label{sec:lower_bounds}
As for many problems handling high dimensional inputs, the best lower bound known for this problem is $\Omega(n \log n)$~\cite{Fredman78}, which derives from the fact that the \textsc{Depth Distribution} is a generalization of the \textsc{Klee's Measure} problem. This bound, however, is tight only when the input is a set of intervals (i.e, $d=1$).
For higher dimensions, the conjectured lower bound of $\Omega(n^{d/2})$ described by Chan~\cite{Chan08} in 2008 for the computational complexity of computing the \textsc{Klee's Measure} can be extended analogously  to the computation of the \textsc{Depth Distribution}.

One intriguing question is whether in dimension $d=2$, as for \textsc{Klee's Measure}, the \textsc{Depth Distribution} can be computed in time within $\bigo{n \log n}$.
We argue that doing so would imply breakthrough results in a long standing problem, \textsc{Matrix Multiplication}. We show that any instance of \textsc{Matrix Multiplication} can be solved using an algorithm which computes the \textsc{Depth Distribution} of a set of rectangles in the plane. For this, we make use of the following simple observation:
\paragraph{\textbf{\emph{Observation 1.}}}
	Let  \emph{A,B} be two $n \times n$ matrices of real numbers, and let $C_i$ denote the $n \times n$ matrix that results from multiplying the $n \times 1$ vector corresponding to the $i$-th column of $A$ with the $1\times n$ vector corresponding to the $i$-th row of $B$. Then, $AB = \sum_{i=1}^{n} C_i$. 

\medskip
We show in \Cref{theo:ddist_matrixmult} that multiplying two $n \times n$ matrices can be done by transforming the input into a set of $\bigo{n^2}$ axis-aligned rectangles, and computing the \textsc{Depth Distribution} of the resulting set. Moreover, this transformation can be done in linear time, thus, the theorem yields a conditional lower bound for the computation of the \textsc{Depth Distribution}.

\begin{theorem}\label{theo:ddist_matrixmult}
Let  $A,B$ be two $n \times n$ matrices of non-negative real  numbers. There is a set $\mathcal{B}$ of rectangles of size within $O(n^2)$, and a domain rectangle $\Gamma$, such that the \textsc{Depth Distribution} of $\mathcal{B}$ within $\Gamma$  can be projected to obtain the value of the product $AB$.
\end{theorem}

\begin{SHORT}
Intuitively, we create a \emph{gadget} to represent each matrix $C_i$. Within the $i$-th gadget, there will be a rectangular region for each component of $C_i$ with the value of that component as volume.
We arrange the boxes so that two distinct regions have the same depth if and only if they represent the same respective coefficients of two distinct matrices $C_i$ and $C_{i'}$ (formally, they represent coefficients $(C_i)_{j,k}$ and $(C_{i'})_{j',k'}$, respectively, such that $i \neq i', j = j'$, and $k = k'$). 

Due to lack of space, we defer the complete proof to the extended version~\cite{2017-ARXIV-DepthDistributionInHighDimension-BabrayPerezRojas}.
\end{SHORT}
\begin{LONG}
\begin{proof}
	
	We create a \emph{gadget} to represent each $C_i$. Within the $i$-th gadget, there will be a rectangular region for each coefficient of $(C_i)$ with the value of that coefficient as volume (see \Cref{fig:matrixMultToDdist1} for a general outlook of the instance).  We arrange the boxes so that two of such rectangular regions have the same depth if and only if they represent the same respective coefficients of two distinct matrices $C_i$ and $C_{i'}$ (formally, they represent coefficients $(C_i)_{j,k}$ and $(C_{i'})_{j',k'}$, respectively, such that $i \neq i', j = j'$, and $k = k'$).
	\begin{figure}
		\begin{center}
			\includegraphics[page=2,width=1\linewidth]{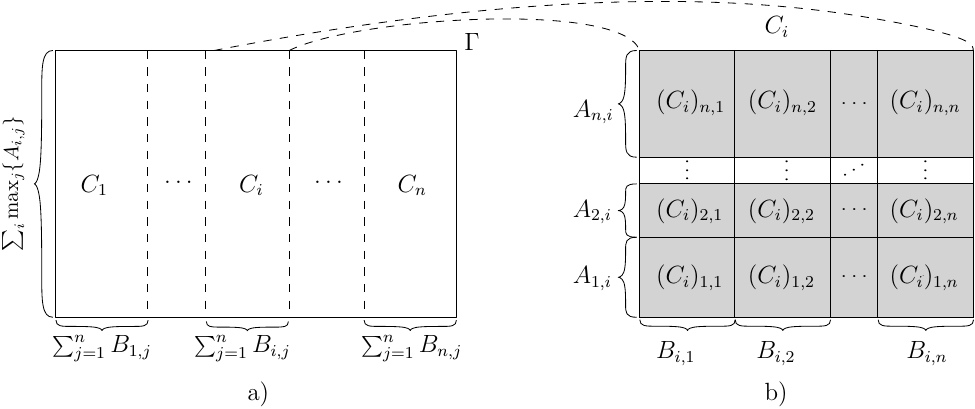}
		\end{center}   
		\caption{An outlook of the instance generated for the product $AB$: we add a gadget for each $C_1, \ldots, C_n$, within a domain $\Gamma$ as in a). In b), a representation of $C_i$ with $2n$ boxes, the volume of the rectangular regions correspond to the coefficients of $C_i$ (the regions in a gadget must have distinct depths to avoid that their volumes are added into a same component of the \textsc{Depth Distribution}).}
		\label{fig:matrixMultToDdist1}
	\end{figure}
	
	We describe a set of $5n^2$ boxes $\mathcal{B}$ (one box for each of the $n^2$ coefficients of $A$, two boxes for each of the $n^2$ coefficients  of $B$, and $2n^2$ additional boxes) such that, for each $i,j \in [1..n]$, the  $(2ni + 2j + 1)$-th component of the \textsc{Depth Distribution} of $\mathcal{B}$ is equal to the component $AB_{i,j}$ of the product $AB$. Such a set can be constructed as follows (see \Cref{fig:matrixMultToDdist2} for a graphical representation of the instance generated):
	
	\begin{itemize}
		\item{Let the domain $\Gamma = \{(x,y) \mid 0 \le x \le \sum_{i}{\sum_{j}{B_{i,j}}}, 0 \le y \le \sum_i{\max_j \{A_{i,j}\}}\}$.} 		
		\item{For all $i \in [1..n]$ we create a gadget for $C_i$ that covers the entire domain in the $y$-direction, and that spans from $C_i^{start} = \sum_{j=1}^{i-1}{\sum_{k=0}^{n}{B_{j,k}}}$} to $C_i^{end} = C_i^{start} + \sum_{k=0}^{n}{B_{i,k}}$ in the $x$-direction.
		\item {Within the gadget for $C_i$ we place one box for each $A_{j,i}$ and two boxes for each $B_{i,j}$, for $i,j \in [1.. n]$, as follows: the one corresponding to $A_{j,i}$ will span $C_i$ entirely in the $x$-direction, and is bounded by $(\sum_{k=1}^{j}{\max_{l=1}^{n} \{A_{k,l}\}}) \le y \le (A_{j,i} + \sum_{k=1}^{j}{\max_{l=1}^{n} \{A_{k,l}\}})$ in the $y$-direction.
		 For $B_{i,j}$ we place two identical boxes entirely spanning $C_i$ in the $y$-direction, and in the $x$-direction bounded by $(C_i^{start} + \sum_{k=1}^{j-1}{B_{i,k}}) \le x \le C_i^{end}$.}
		
		\item{Finally, we add $2n^2$ boxes to ensure that rectangular regions corresponding to two coefficients $C_{i,j}$ and $C_{i,k}$ in distinct rows $j,k$ of a same $C_i$ do not share the same depth, for all $i,j,k \in [1.. n]$ . For this, for all $j \in [1..n]$ we add $2n$ identical boxes entirely spanning the domain   in the $x$-direction,  and spanning from $(\sum_{k=1}^{j}{\max_{l=1}^{n} \{A_{k,l}\}})$} to $(\sum_i{\max_j \{A_{i,j}\}}))$ in the $y$-direction.
	\end{itemize}
	Note that in the instance generated, for $i,j \in [1..n]$:
	\begin{itemize}
		\item a region has odd depth if and only if its volume is equal to some coefficient of any $C_i$;
		\item the regions corresponding to coefficients of the $i$-th rows have depth between $(2in + 3)$ and $(4in+1)$;
		\item within the gadget for each $C_i$, the rectangular region with volume corresponding to the coefficient $C_{i,j}$ has depth $(2ni + 2j + 1)$, and no other rectangular region within the gadget has that depth;
		\item two distinct regions have the same depth if and only if they represent the same respective coefficients of two distinct matrices $C_i$ and $C_{i'}$.
	\end{itemize}
	 The arguments above and the fact that, by definition of the \textsc{Depth Distribution} the volumes of regions with the same depth are accumulated together, yield the result of the theorem.
	\qed
\end{proof}
\end{LONG}

The optimal time to compute the product of two $n \times n$ matrices is still open. It can naturally be computed in time within $O(n^3)$. However, Strassen showed  in 1969  that within $O(n^{2.81})$ arithmetic  operations  are enough~\cite{Strassen1969}. This gave rise to a new area of research, where the central question is to determine the value of the exponent of the computational complexity of square matrix multiplication, denoted $\omega$, and defined  as  the minimum  value  such  that  two $n \times n$ matrices  can be  multiplied  using within $O(n^{\omega + \varepsilon})$ arithmetic operations for any $\varepsilon > 0$. 

\begin{LONG}
\begin{landscape}
	\begin{figure}[h]
		\begin{center}
			\includegraphics[page=1, width=.85\linewidth]{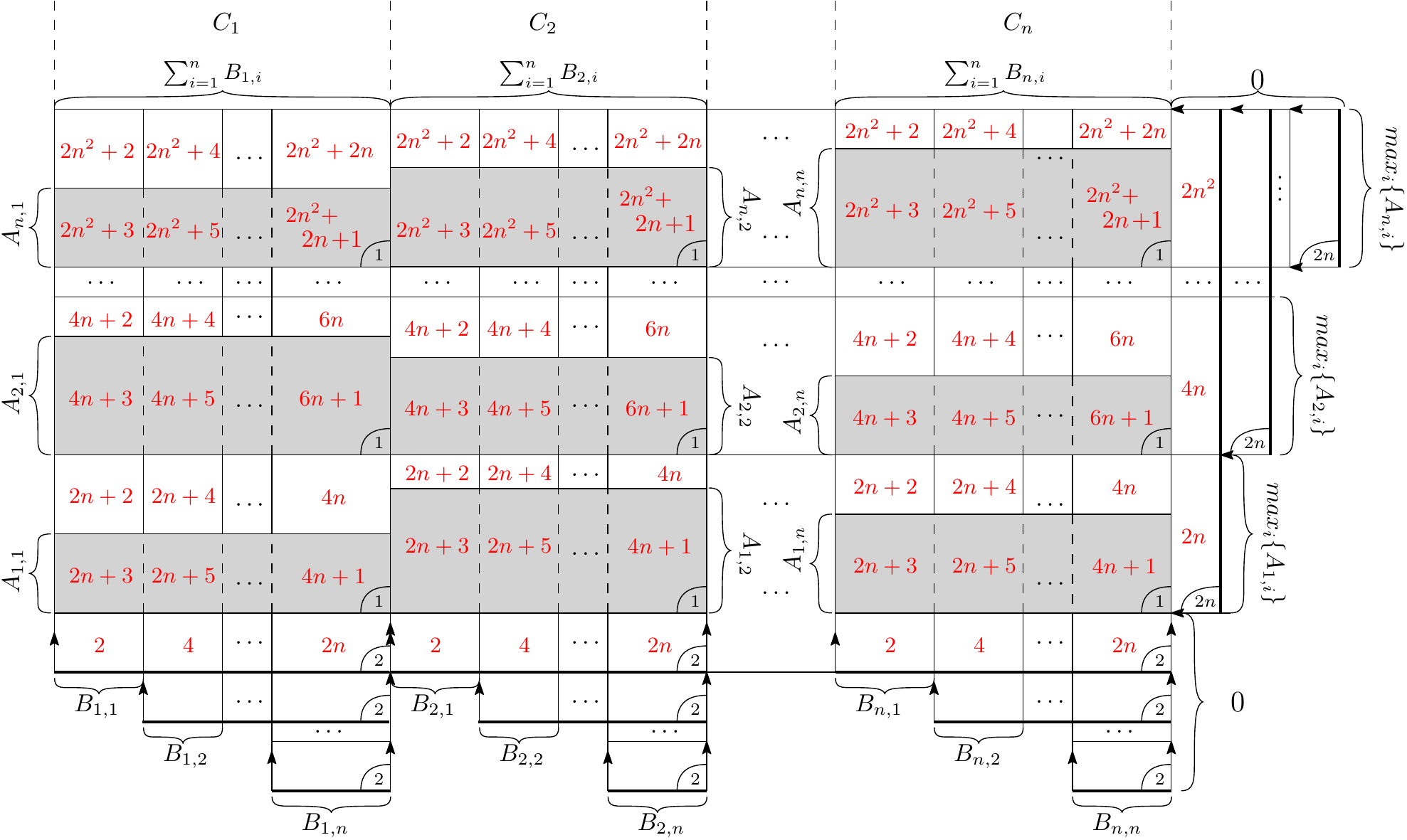}
		\end{center}   
		\caption{Illustration of an instance of \textsc{Depth Distribution} generated for the product $AB$.  The text in red is the depth of the region. The small arrows indicate that the boxes they delimit span the entire domain in the direction they point to. Small numbers in the corner of each box indicate the number of exact copies of that box added to the instance (or intuitively, the weight of the box). Finally, the numbers over curly brackets indicate the length of the region delimited by the brackets.
		}
		\label{fig:matrixMultToDdist2}
	\end{figure}
\end{landscape}
\end{LONG}

The result of \Cref{theo:ddist_matrixmult} directly yields a conditional lower bound on the complexity of \textsc{Depth Distribution}: in particular, \textsc{Depth Distribution} in dimension as low as two, can be solved in time within $\bigo{n \log n}$, then \textsc{Matrix Multiplication} can be computed in time within $\bigo{n^2}$, i.e. $\omega = 2$. However, this would be a great breakthrough in the area, the best known upper bound to date is approximately $\omega \le 2.37$, when improvements in the last 30 years~\cite{CoppersmithW87,Gall14} have been in the range $[2.3728, 2.3754]$.

\begin{corollary} [Conditional lower bound]\label{col:ddist2d_lowerbound}
	Computing the \textsc{Depth Distribution} of a set $B$ of $n$ $d$-dimensional boxes requires  time within $\Omega(n^{1+c})$, for some constant $c > 0$, unless  two $n \times n$ matrices can be multiplied in time $\bigo{n^{2+\varepsilon}}$, for any constant $\varepsilon > 0$.
\end{corollary}

The running time of the algorithm that we described in \Cref{theo:ddist_comp_spacey} can be improved for large classes of instances (i.e. asymptotically infinite) by considering measures of the difficulty of the input other than its size. We describe two of these improved solutions in the next section.

\section{Multivariate analysis}\label{sec:finer_analysis}
Even though the asymptotic complexity of $\bigo{n^{\frac{d+1}{2}} \log n}$  is the best we know so far for computing the \textsc{Depth Distribution} of a set of $n$ $d$-dimensional boxes, there are many cases which can be solved faster.
Some of those ``easy'' instances can be mere particular cases, but others can be hints of some hidden measures of difficulty of the \textsc{Depth Distribution} problem.
We show that, indeed, there are at least two such difficulty measures, gradually separating instances of the same size $n$ into various classes of difficulty. Informally, the first one (the \emph{profile} of the input set, Section~\ref{sec:profile}) measures how separable the boxes are by axis-aligned hyperplanes, whereas the second one (the \emph{degeneracy} of the intersection graph, Section~\ref{sec:inters-graph-degen}) measures how ``complex" the interactions of the boxes are in the set between them. Those measures inspire similar results for the computation of the \textsc{Klee's Measure} and of the \textsc{Maximum Depth}.

\subsection{Profile} \label{sec:profile}
The \emph{$i$-th profile} $p_i$ of a set of boxes $\mathcal{B}$ is the maximum number of boxes intersected by any hyperplane orthogonal to the $i$-th dimension; and the \emph{profile} $p$ of $\mathcal{B}$ is the minimum $p=\min_{i\in[1..d]}\{p_i\}$ of those over all dimensions.
D'Amore~\cite{dAmoreNRW95} showed how to compute it in linear time (after sorting the coordinate of the boxes in each dimension).
The following lemma shows that the \textsc{Depth Distribution} can be computed in time sensitive to the profile of the input set.
%

\begin{lemma}\label{lem:profile_ddist}
Let $\mathcal{B}$ be a set of boxes with profile $p$, and $\Gamma$ be a $d$-dimensional axis-aligned domain box. The \textsc{Depth Distribution} of $\mathcal{B}$ within $\Gamma$ can be computed in time within $\bigo{n \log n + np^{\frac{d-1}{2}}\log p} \subseteq \bigo{n^\frac{d+1}{2}\log n}$.
\end{lemma}
\begin{SHORT}
Due to lack of space we defer the complete proof to the extended version~\cite{2017-ARXIV-DepthDistributionInHighDimension-BabrayPerezRojas}.
\end{SHORT}
\begin{LONG}
\begin{proof}
We describe an algorithm which partitions the domain into independent slabs, computes the \textsc{Depth Distribution} within each slab, and combines the results into the final answer.
For this, it sweeps a plane by the dimension with smallest profile, and after every $2p$ endpoints, it creates a new slab cutting the space with a hyperplane orthogonal to this dimension.
This partitions the space into $\bigo{n/p}$ slabs, each intersecting at most $\bigo{p}$ boxes.
Finally, the algorithm computes the \textsc{Depth Distribution} of $\mathcal{B}$ within each slab in time within $\bigo{p^\frac{{d+1}}{2} \log p}$, and obtains the \textsc{Depth Distribution} of $\mathcal{B}$ within $\Gamma$ by summing the respective components of the \textsc{Depth Distribution} within each slab.
The total running time of this is within $\bigo{n \log n + np^{\frac{d-1}{2}}\log p}$.
\qed
\end{proof}
\end{LONG}

The lemma above automatically yields refined results for the computation of the \textsc{Klee's Measure} and the \textsc{Maximum Depth} of a set of boxes $\mathcal{B}$. However, applying the technique in an \emph{ad-hoc} way to these problems yields better bounds:

\begin{corollary}\label{col:profile_km}
	Let $\mathcal{B}$ be a set of boxes with profile $p$, and $\Gamma$ be a domain box. The \textsc{Klee's Measure} and \textsc{Maximum Depth} of $\mathcal{B}$ within $\Gamma$ can be computed in time within $\bigo{n \log n + np^{\frac{d-2}{2}}} \subseteq \bigo{n^{d/2}}$.
\end{corollary}

The algorithms from \Cref{lem:profile_ddist} and \Cref{col:profile_km} asymptotically outperform previous ones  in the sense that their running time is never worse than previous algorithms by more than a constant factor, but can perform faster by more than a constant factor on specific families of instances.

An orthogonal approach is to consider how complex the interactions between the boxes are in the input set $\mathcal{B}$, analyzing, for instance, the  intersection graph of $\mathcal{B}$. We study such a technique in the next section.

\subsection{Intersections Graph Degeneracy}\label{sec:inters-graph-degen}
A \emph{$k$-degenerate} graph is an undirected graph in which every subgraph has a vertex of degree at most $k$~\cite{Lick1970}.
Every $k$-degenerate graph accepts an ordering of the vertices in which every vertex is connected with at most $k$ of the vertices that precede it (we refer below to such an ordering as a \emph{degenerate ordering}).

In the following lemma we show that this ordering can be used to compute the \textsc{Depth Distribution} of a set $\mathcal{B}$ of $n$ boxes in running time sensitive to the degeneracy of the intersection graph of $\mathcal{B}$.

\begin{lemma}
	Let $\mathcal{B}$ be a set of boxes and $\Gamma$ be a domain box, and let $k$ be the degeneracy of the intersection graph $G$ of the boxes in $\mathcal{B}$. The \textsc{Depth Distribution} of $\mathcal{B}$ within $\Gamma$ can be computed in time within  $\bigo{n \log^d n + e+ nk^{\frac{d+1}{2}}}$, where $e \in \bigo{n^2}$ is the number of edges of $G$.
\end{lemma}
\begin{proof}
	We describe an algorithm that runs in time within the bound in the lemma.
	The algorithm first computes the intersection graph $G$ of $\mathcal{B}$ in time within $\bigo{n \log^d n + e}$~\cite{Edelsbrunner1983}, as well as the $k$-degeneracy of this graph and a degenerate ordering $O$ of the vertices in time within $\bigo{n+e}$~\cite{Matula1983}.
	The algorithm then iterates over $O$ maintaining the invariant that, after the $i$-th step, the \textsc{Depth Distribution} of the boxes corresponding to the vertices $v_1,v_2,\ldots, v_i$ of the ordering has been correctly computed.
	
	For any subset $V$ of vertices of $G$, let $DD_\mathcal{B}^\Gamma(V)$ denote the \textsc{Depth Distribution} within $\Gamma$ of the boxes in $\mathcal{B}$ corresponding to the vertices in $V$. Also, for $i \in [1..n]$ let $O[1..i]$ denote the first $i$ vertices of $O$, and $O[i]$ denote the $i$-th vertex of $O$. From $DD_\mathcal{B}^\Gamma(O[1..i\text{-1}])$ (which the algorithm ``knows" after the ($i$-1)-th iteration), $DD_\mathcal{B}^\Gamma(O[1..i])$ can be obtained as follows:
	($i.$) let $P$ be the subset of $O[1..i\text{-1}]$ connected with $O[i]$; ($ii.$) compute $DD_\mathcal{B}^{O[i]}(P \cup \{O[i]\})$ in time within $\bigo{k^\frac{d+1}{2}\log k}$ using \texttt{SDC-DDistribution} (note that the domain this time is $O[i]$ itself, instead of $\Gamma$);
	($iii.$) add to $(DD_\mathcal{B}^\Gamma(O[1..i]))_1$ the value of $(DD_\mathcal{B}^{O[i]}(P \cup O[i]))_1$;
	and ($iv.$) for all $j=[2..k\text{+1}]$, substract from  $(DD_\mathcal{B}^\Gamma(O[1..i]))_{j-1}$ the value of $(DD_\mathcal{B}^{O[i]}(P \cup O[i]))_j$ and add it to $(DD_\mathcal{B}^\Gamma(O[1..i\text{-1}]))_j$.
	
	Since the updates to the \textsc{Depth Distribution} in each step take time within $\bigo{k^\frac{d+1}{2}\log k}$, and there are $n$ such steps, the result of the lemma follows.
	\qed
\end{proof}

Unlike the algorithm sensitive to the profile, this one can run in time within $\bigo{n^{1+\frac{d+1}{2}}}$ (e.g. when $k=n$), which is only better than the $\bigo{n^\frac{d+1}{2}}$ complexity of \texttt{SDC-DDistribution} for values of the degeneracy $k$ within $\bigo{n^{1-{\frac{2}{d}}}}$.

Applying the same technique to the computation of \textsc{Klee's Measure} and \textsc{Maximum Depth} yields improved solutions as well:

\begin{corollary}
	Let $\mathcal{B}$ be a set of boxes and $\Gamma$ be a domain box, and let $k$ be the degeneracy of the intersection graph $G$ of the boxes in $\mathcal{B}$. The \textsc{Klee's Measure} and \textsc{Maximum Depth} of $\mathcal{B}$ within $\Gamma$ can be computed in time within  $\bigo{n \log^d n + e+ nk^{\frac{d}{2}}}$, where $e \in \bigo{n^2}$ is the number of edges of $G$.
\end{corollary}

Such refinements of the worst-case complexity analysis are only examples and can be applied to many other problems handling high dimensional data inputs. We discuss a selection in the next section.

\begin{INUTILE}
\subsection{Computing only selected elements of the Depth Distribution}\label{sec:comp-only-select}
An intriguing question is whether the computation of a particular $V_i$ can be performed in time depending on $i$.
For instance, $V_n$ can be computed in linear time since it is the volume of the intersection of all the boxes. Besides, for any constant $c$, $V_{n-c}$ can be computed in time within $O(n^c)$. However, we show below that in many cases computing $V_i$ requires as many time as solving the \textsc{Coverage} problem, which detects whether a set of boxes $\mathcal{B}$ covers a given domain $\Gamma$. 

\begin{lemma}
Let $\mathcal{B}$ be a set of boxes and $\Gamma$ be a domain box. Computing the volume of the region within $\Gamma$ covered by exactly $n/f(n)$ boxes from $\mathcal{B}$, for any function $f$ such that $1 < f(n) \le n$, is as hard as computing whether $\mathcal{B}$ \textsc{Coverage} can be solved in that time.	
\end{lemma}
\begin{proof}
Construct $\mathcal{B}'$ by adding $\frac{n}{c-1}$ boxes to $\mathcal{B}$ covering the entire domain. The size $n'$ of $\mathcal{B}'$ is $n'=\frac{cn}{c-1}$, and the fraction of boxes added with respect to $n'$ is $n'/c$. Since the boxes added cover completely the domain, $\mathcal{B}$ covers the domain if and only if the $(n'/c)$-th element of the \textsc{Depth Distribution} of $\mathcal{B}'$ is zero.
\qed
\end{proof}

The above lemma allows to argue that potentially improved running times based on being sensitive to the index $i$ of the selected component will still be within $\Omega(n^{d/2})$, for many values of $i$ (except maybe for $i$ within $n-\emph{o}(n)$).

\begin{TODO}
	Add running time base on $i$ which improves the running time within the range $[n^{d/2}, n^{(d+1)/2}]$
\end{TODO}

\end{INUTILE}

\section{Discussion}\label{sec:discussion}		

The \textsc{Depth Distribution} captures many of the features in common between \textsc{Klee's measure} and \textsc{Maximum Depth}, so that new results on the computation of the \textsc{Depth Distribution} will yield corresponding results for those two measures, and has its own applications of interest. 
Nevertheless, there is no direct reduction to \textsc{Klee's measure} or \textsc{Maximum Depth} from \textsc{Depth Distribution}, as the latter seems computationally more costly, and clarifying further the relationship between these problems will require finer models of computation.
We discuss below some further issues to ponder about those measures.

\noindent\emph{Discrete variants.}
\label{sec:discreteVariants}
In practice, multidimensional range queries are applied to a database of multidimensional points. 
This yields discrete variants of each of the problems previously discussed~\cite{AboKhamis2015,YildizHershbergerSuri11}. 
In the \textsc{Discrete Klee's measure}\begin{TODO}
Add here alternative names introduced by ~\cite{AboKhamis2015} and \cite{YildizHershbergerSuri11}
\end{TODO}, the input is composed of not only a set $\cal B$ of $n$ boxes, but also of a set $S$ of $m$ points. The problem is now to compute not the volume of the union of the boxes, but the number (and/or the list) of points which are covered by those boxes.
\begin{TODO}
Describe here in two sentences each the work of
 \textcite{AboKhamis2015} and 
 \textcite{YildizHershbergerSuri11}
\end{TODO}
Similarly, one can define a discrete version of the \textsc{Maximum Depth} (which points are covered by the maximum number of boxes) and of the \textsc{Depth Distribution} (how many and which points are covered by exactly $i$ boxes, for $i\in[1..n]$).
Interestingly enough, the computational complexity of these discrete variants is much less than that of their continuous versions when there are reasonably few points~\cite{YildizHershbergerSuri11}: the discrete variant becomes hard only when there are many more points than boxes~\cite{AboKhamis2015}.
Nevertheless, ``easy'' configurations of the boxes also yield ``easy'' instances in the discrete case: it will be interesting to analyze the discrete variants of those problems according to the measures of \emph{profile} and \emph{k-degeneracy} introduced on the continuous versions. 
\begin{TODO}

\noindent\emph{Finer Multivariate Analysis.}
\label{sec:furtherRefinement}
Whereas we refined the worst case complexity over instances of fixed size by introducing and taking account various parameters, there is still room for further refinements.
Javiel: add here some of the open questions

\end{TODO}

\noindent\emph{Tighter Bounds.}
\label{sec:improvedLowerBounds}
Chan~\cite{Chan2013} conjectured that a complexity of $\Omega(n^{d/2})$ is required to compute the \textsc{Klee's Measure}, and hence to compute the \textsc{Depth Distribution}.
However, the output of \textsc{Depth Distribution} gives much more information than the \textsc{Klee's Measure}, of which a large part can be ignored during the computation of the \textsc{Klee's Measure} (while it is required for the computation of the \textsc{Depth Distribution}).  It is not clear whether even a lower bound of $\Omega(n^{d/2 + \epsilon})$ can be proven on the computational complexity of the \textsc{Depth Distribution} given this fact.
\begin{TODO}
JAVIEL: Complete here.
\end{TODO}


\medskip 
\noindent\textbf{Funding:} All authors were supported by the Millennium Nucleus ``Information and Coordination in Networks'' ICM/FIC RC130003. J\'er\'emy Barbay and Pablo P\'erez-Lantero were supported by the projects CONICYT Fondecyt/Regular nos 1170366 and 1160543 (Chile) respectively, while Javiel Rojas-Ledesma was supported by CONICYT-PCHA/Doctorado Nacional/2013-63130209 (Chile).

\printbibliography




\end{document}